\documentclass[10pt]{article}
\usepackage{enumitem}
\usepackage{amsmath}
\usepackage{amsfonts}
\usepackage{subfigure}
\usepackage{graphicx}
\usepackage[paper=a4paper,textwidth=170mm,textheight=260mm]{geometry}
\usepackage{algorithm}
\usepackage{algorithmic}
\usepackage{listings}
\usepackage[usenames,dvipsnames]{color}
\usepackage{amsthm}
\usepackage[width=12cm,format=plain,labelfont=sf,bf,small,textfont=sf,up,small,justification=justified,labelsep=period]{caption}
\usepackage{amssymb}

\newcommand{\argmax}{\operatornamewithlimits{argmax}}
\newtheorem{thm}{Theorem}[section] 
\newtheorem{lem}{Lemma}[section]
\newtheorem{cor}{Corrolary}[section]

\renewcommand{\algorithmicrequire}{\textbf{Input:}}
\renewcommand{\algorithmicensure}{\textbf{Output:}}

\begin{document}
\title{Thresholding-based reconstruction of compressed correlated signals}
\author{Alhussein Fawzi, Tamara To\v{s}i\'{c} and Pascal Frossard}
\date{}
\maketitle

\begin{abstract}
We consider the problem of recovering a set of correlated signals (\emph{e.g.,} images from different viewpoints) from a few linear measurements per signal. We assume that each sensor in a network acquires a compressed signal in the form of linear measurements and sends it to a joint decoder for reconstruction. We propose a novel joint reconstruction algorithm that exploits correlation among underlying signals. Our correlation model considers geometrical transformations between the supports of the different signals. The proposed joint decoder estimates the correlation and reconstructs the signals using a simple thresholding algorithm. We give both theoretical and experimental evidence to show that our method largely outperforms independent decoding in terms of support recovery and reconstruction quality.
\end{abstract}

\section{Introduction}\label{sec:intro}

The growing number of distributed systems in recent years has led to an important body of work on the efficient representation of signals captured by multiple sensors. Recently, ideas based on Compressed Sensing (CS) \cite{Donoho06, Candes05} have been applied to distributed reconstruction problems \cite{Duarte05} in order to recover signals from a few measurements per sensor. 
When signals are correlated, a joint decoder that properly exploits the inter-sensor dependencies is expected to outperform independent decoding in terms of reconstruction quality. 
Very often, the correlation model restricts the unknown signals to share a common support. Using this correlation model, the authors in \cite{Duarte05} and \cite{Golbabaee09} propose decoding algorithms and show analytically that joint reconstruction outperforms independent reconstructions. In many applications, this correlation model is however too restrictive. For example, in the case of a network of neighbouring cameras capturing one scene or seismic signals captured via different sismometers, the supports of the signals are quite different even if components are linked by simple transformations. 

In this paper, we adopt a more general correlation model and build a joint decoder that recovers the unknown signals from a few measurements per sensor. We assume that the unknown signals are sparse in a \emph{redundant dictionary} $\mathcal{D}$, and not necessarily in an orthonormal basis \cite{Rauhut08,Candes10}. We denote the components of the dictionary as \emph{atoms}. 
We assume that the support of each view $j$ is related to the support of a reference view by a transformation $T_j^*$. 
The transformation $T_j^*$ could be for example a translation function.  Using the given correlation model, we build a joint decoder based on the thresholding algorithm \cite{Rauhut08} and prove theoretically that it outperforms an independent decoding method in terms of recovery rate. Moreover, we show experimentally that the proposed algorihm leads to better reconstruction quality.
\section{Problem formulation}\label{sec:problem}
We consider a sensor network of $J$ nodes. Each sensor $j$ acquires $M$ linear measurements of the unknown signal $y_j \in \mathbb{R}^N$ ($M < N$) and sends it to a central decoder . The role of the decoder is to estimate the unknown signals $\mathcal{Y} = \{y_j\}_{j=1}^J$. By denoting $\mathcal{S} = \{s_j\}_{j=1}^J$ the set of compressed signals acquired by the sensors and $\mathcal{A} = \{A_j\}_{j=1}^{J}$ the sensing matrices, we have:
\begin{equation}
\underbrace{s_j}_{M \times 1} = \underbrace{A_j}_{M \times N} \underbrace{y_j}_{N \times 1}.
\end{equation}
In the rest of this paper, we use independent sensing matrices with Gaussian i.i.d entries. Specifically, $\sqrt{M} (A_j)_{m, n}$ follows a standard Gaussian distribution, for any $m,n,j$.

We assume that the unknown signals $y_j \in \mathbb{R}^N$ are sparse in some dictionary $\mathcal{D}$ that consists of $K$ atoms and denote by $\Phi = [\varphi_1 | \dots | \varphi_K]$ its matrix representation. Formally, we have $y_j = \Phi c_j$, where $c_j$ is a vector of length $K$ with at most $S$ non zero components and $S < N$. 
By denoting the support of $y_j$ with $\Delta_j^*$ (\emph{i.e.,} the set of $S$ atoms corresponding to the non zero entries of $c_j$), $y_j$ can be written as follows:
\begin{equation}
\label{signal_model}
\underbrace{y_j}_{N \times 1} = \underbrace{\Phi_{\Delta_j^*}}_{N \times S} \underbrace{x_j}_{S \times 1},
\end{equation} 
where $\Phi_{\Delta_j^*}$ is the restriction of $\Phi$ to $\Delta_j^*$ and $x_j$ corresponds to the non zero entries of $c_j$.

We adopt the following correlation model: for any $j \in \{1, \dots, J\}$, the set of atoms in $\Delta_j^*$ can be obtained from $\Delta_1^*$ by applying a transformation $T_j^*:\mathcal{D}\rightarrow\mathcal{D}$. This can be written as: $T_j^*(\Delta_1) = \Delta_j^*$ (we consider that $T_1^*$ is the identity). 
In our problem, the vector of transformations $T^* = \{T_j^*\}_{j=1}^J$ is unknown. 
However, we assume that we are given a finite set $\mathcal{T}$ of candidate transformations vectors and that the correct vector $T^*$ belongs to $\mathcal{T}$.


Considering the above correlation model, we address the following problem: Given the compressed signals $\mathcal{S}$, the sensing matrices $\mathcal{A}$, the sparsity $S$, the dictionary $\Phi$, and the set of candidate transformations vectors $\mathcal{T}$, estimate the unknown signals $\mathcal{Y}$ (\emph{i.e.,} supports $\{\Delta_j^*\}_{j=1}^J$ and coefficients $\{x_j\}_{j=1}^J$) using a small number of measurements per sensor $M$. 

\section{Joint thresholding algorithm}\label{sec:JT}
We propose a solution to the problem formulated in the previous section.
Our proposed decoder extends the simple thresholding algorithm \cite{Rauhut08} to multiple signals. This choice is motivated by the low complexity of thresholding algorithm with respect to other decoding methods \cite{XChen09}. Our joint decoder represents an efficient alternative when the signals are simple (\emph{i.e.,} they have very sparse representations in the dictionary) and the number of sensors is fairly large, so that other decoding methods become computationally intractable.

The Joint Thresholding (JT) decoder exploits the information diversity brought by the different signals to reduce the number of measurements per sensor required for accurate signals reconstruction. It groups the measurements obtained from each individual signal and precisely estimates the unknowns $(\Delta_1^*, T^*)$ (or equivalently all the supports $\{\Delta_j^*\}_{j=1}^{J}$).

JT obtains an estimate $(\widehat{\Delta_1}, \widehat{T})$ of $(\Delta_1^*, T^*)$ by maximizing the following objective function, which is called the score function:
\begin{equation}
\label{scoreFunction}
\Psi_s(\Delta_1, T) = \sum\limits_{j=1}^{J} \sum\limits_{\varphi \in T_j(\Delta_1)} s_j \cdot A_j \varphi,
\end{equation}
where $\Delta_1$ and $T$ denote respectively the support of the reference signal and the vector of transformations variables and the operator $\cdot$ denotes the canonical inner product. 
The use of $\Psi_s$ as the objective function is justified by:
\begin{align*}
\vspace{-1mm}
(\widehat{\Delta_1}, \widehat{T}) & = \argmax_{(\Delta_1, T)} \Psi_s(\Delta_1, T) \\ & \approx \argmax_{(\Delta_1, T)} \Psi_y(\Delta_1, T) \\ & = (\Delta_1^*, T^*),
\vspace{-1mm}
\end{align*}
where $\Psi_y(\Delta_1, T) = \sum_{j=1}^J \sum_{\varphi \in T_j(\Delta_1)} y_j \cdot \varphi = \mathbb{E} \Psi_s$ ($\mathbb{E}$ denotes the expected value of a random variable).
Indeed, if both assumptions given in Eq.(\ref{assumptionThresholding}) and Eq.(\ref{positivity}) hold\footnote{The assumptions are discussed in the next section.},
$\Psi_y(\Delta_1, T)$ is maximal for $\Delta_1 = \Delta_1^*$ and $T = T^*$. Besides, for large values of $M$, $\Psi_s(\Delta_1, T)$ concentrates around its average value $\Psi_y(\Delta_1, T)$ (Lemma \ref{fundLemma}).
The description of JT algorithm is given in Algorithm \ref{alg:jt}.
\begin{algorithm}[ht]
\caption{Joint Thresholding (JT) Algorithm}
\algorithmicrequire{\hspace{1mm} compressed signals $\{s_j\}$, sensing matrices $\{A_j\}$, sparsity $S$, dictionary $\Phi$, candidate vectors of transformations $\mathcal{T}$} \\
\algorithmicensure{\hspace{1mm} estimated signals $\{\widehat{y_j}\}$, support $\widehat{\Delta_1}$ and vector of transformations $\widehat{T}$}
\label{alg:jt}
\begin{algorithmic}
\STATE $\mathbf{1.}$ Initialization: $(\widehat{\Delta_1}, \widehat{T}, \widehat{\Psi}) \leftarrow (\varnothing, \varnothing, -\infty)$.
\STATE $\mathbf{2.}$ For every $T \in \mathcal{T}$
\STATE \hspace{2mm}$\mathbf{2.1}$ Build the vector $d_T$ of length $K$ in the following way:
	\begin{equation*}
		d_T = \sum\limits_{j=1}^J (A_j T_j (\Phi))^\mathrm{T} s_j,
	\end{equation*}
\STATE	where $T_j (\Phi) = \begin{bmatrix} T_j(\varphi_1) \ldots T_j(\varphi_K) \end{bmatrix}$.
\STATE \hspace{2mm}$\mathbf{2.2}$ Keep the largest $S$ entries in $d_{T}$ and set the other entries to zero. The positions of the non zero entries in $d_{T}$ give the indices of the estimated support $\Delta_1$ of the first signal.
\STATE \hspace{2mm}$\mathbf{2.3}$ Calculate the \emph{score} $\Psi_s(\Delta_1, T)$ by summing the $S$ non zero entries of $d_{T}$.
\STATE \hspace{2mm}$\mathbf{2.4}$ If $\Psi_s(\Delta_1, T)$ exceeds $\widehat{\Psi}$: update $(\widehat{\Delta_1}, \widehat{T}, \widehat{\Psi}) \leftarrow (\Delta_1, T, \Psi_s(\Delta_1, T))$
\STATE $\mathbf{3.}$ Build the coefficients vector for all $j \in \{1, \dots, J\}$:
	\begin{equation*}
		\widehat{x_j} = \left( A_j \Phi_{ \widehat{\Delta_j} } \right)^{+} s_j.
	\end{equation*}
	where $(\cdot)^{+}$ denotes the pseudo-inverse operator. Note that $\widehat{\Delta_j}$ is obtained using the correlation model: $\widehat{\Delta_j} = \widehat{T_j}(\widehat{\Delta_1})$, for $j \geq 2$.
\STATE $\mathbf{4.}$ Obtain signals estimates:
		\begin{equation*}
		\widehat{y_j} = \Phi_{ \widehat{\Delta_j} } \widehat{x_j}.
	\end{equation*}
\end{algorithmic}
\end{algorithm}

In words, the JT algorithm calculates for each transformation vector $T \in \mathcal{T}$ the vector $d_T$, whose entries are given by \mbox{$d_T[i] = \sum\limits_{j=1}^J s_j \cdot A_j T_j (\varphi_i)$}, for $1 \leq i \leq K$. Then, the largest $S$ elements in $d_T$ are summed and assigned to $\Psi_s(\Delta_1, T)$. Estimated quantities $\{\widehat{\Delta_1}, \widehat{T}\}$ are updated if $\Psi_s(\Delta_1, T)$ achieves a higher score. Knowing the set of supports, we deduce coefficients $\widehat{x_j}$ by computing the least squares solution to equation $A_j \Phi_{\widehat{\Delta_j}} \widehat{x_j} = s_j$.

\section{Theoretical analysis}\label{sec:bounds}
\label{quantitative}

Our theoretical analysis focuses on the performance of JT in finding the correct supports.
Hence, we will not address the quality of the estimated coefficients $\{\widehat{x_j}\}_{j=1}^J$.  
In particular, we focus on the analysis of the \emph{recovery rate} $R$ defined as the total number of correctly recovered atoms (in all signals combined) divided by the total number of atoms (\emph{i.e.,} $SJ$). 


We assume the following: 
\begin{itemize}
\item For all $j \in \{1, \dots, J\}$, $\Delta_j^*$ can be recovered entirely by applying the thresholding algorithm on $y_j$.
Formally, there exists $\eta > 0$ verifying:
\begin{equation}
\label{assumptionThresholding}
\inf_{\varphi \in \Delta_j^*} \left| \frac{y_j}{\|y_j\|_2} \cdot \varphi \right| > \sup_{\varphi \in \overline{\Delta_j^*}} \left| \frac{y_j}{\|y_j\|_2} \cdot \varphi \right| + \eta,
\end{equation}
where $\overline{\Delta_j^*}$ is the complement of $\Delta_j^*$ in $\mathcal{D}$.
As this condition is practically hard to verify, a sufficient condition involving the coherence of the dictionary is given in \cite[Eq.(3.2)]{Rauhut08}.
\item All the atoms in the supports have positive inner products with the corresponding signal:
\begin{equation}
\label{positivity}
\forall \varphi \in \Delta_j^*, y_j \cdot \varphi \geq 0
\end{equation}
\end{itemize}
The assumption in Eq.(\ref{assumptionThresholding}) is reasonable since we cannot hope to recover the supports using JT unless the thresholding algorithm correctly recovers the supports when applied on the full signals $y_j$. Assumption in Eq.(\ref{positivity}) is a technical one and is used in the proof of our main theorem. Intuitively, it guarantees that $\Psi_y(\Delta_1, T) = \mathbb{E} \Psi_s(\Delta_1, T)$ is maximal when $(\Delta_1, T) = (\Delta_1^*, T^*)$. This assumption can be achieved by adding the inverse of the atom in the dictionary ($\varphi \rightarrow - \varphi$) when the inner product is negative. 

The main ingredient we will use in our analysis is the concentration of $\frac{1}{J} \sum_{j=1}^J A_j u_j \cdot A_j v_j$ around its average value $\frac{1}{J} \sum_{j=1}^J u_j \cdot v_j$ for any set of vectors $\{u_j\}_{j=1}^J$ and $\{v_j\}_{j=1}^J$ of length $N$. This is shown in the following lemma:

\begin{lem}
\label{fundLemma}
Let $\{u_j\}_{1 \leq j \leq J}$ and $\{v_j\}_{1 \leq j \leq J}$ with $u_j, v_j \in \mathbb{R}^N$, such that $\|u_j\|_2 \leq B_u$ and $\|v_j\|_2 \leq B_v$ for all $j \in \{1, \dots, J\}$.
Assume that $\{A_j\}_{1 \leq j \leq J}$ are independent random matrices of dimension $M \times N$, with iid entries following $\mathcal{N}(0, \frac{1}{M})$.
Then, for all $\tau > 0$,
\begin{align}
\mathbb{P} \left(\frac{1}{J} \left| \sum\limits_{j = 1}^J A_j u_j \cdot A_j v_j - u_j \cdot v_j \right| \geq \tau \right)
\leq 2 \exp \left ( - \frac{J M \tau^2}{C_1 B_u^2 B_v^2 + C_2 \tau B_u B_v} \right).
\end{align}
with $C_1=\frac{8e}{\sqrt{6 \pi}}$ and $C_2=2\sqrt{2}e$.
\end{lem}
The proof of this lemma can be found in Appendix A.

Our main theoretical result is given in the following theorem:

\begin{thm}[Recovery rate of JT]
\label{mainTheorem}
Let $R$ be the recovery rate of JT defined by:
\begin{align*}
R = \frac{\sum_{j=1}^J |\Delta_j^* \cap \widehat{\Delta_j}|}{SJ}.
\end{align*}
Then, for any $0 < \alpha \leq 1$:
\begin{align}
\label{mainResult}
\mathbb{P} \left( R \geq 1 - \alpha \right) \geq 1 - 4 S J K |\mathcal{T}| \exp \left(- C M J \eta^2 \alpha^2 \frac{m_y^2}{M_y^2} \right),
\end{align}
where $m_y = \min_{j} \|y_j\|_2$, $M_y = \max_{j} \|y_j\|_2$, $C = \left(\frac{32 e}{\sqrt{6\pi}} + 4e \sqrt{2} \right)^{-1}$.
\end{thm}
The proof of this theorem can be found in Appendix B.

For simplicity, we consider the common case where all the signals have the same energy ($m_y = M_y$). Theorem \ref{mainTheorem} shows that for sufficiently high values of $J$, the recovery rate is mainly governed by $MJ$, $\eta$ and $|\mathcal{T}|$. The dependence on $MJ$ (\emph{i.e.,} total number of measurements) follows our intuition as JT combines the measurements of the different sensors to perform the joint decoding. Increasing the total number of measurements leads to a better recovery rate. The quantity $\eta$ hides the dependence of $R$ on the signal characteristics and model. For clarification, the following inequality provides a lower bound on $\eta$, in terms of sparsity, coherence of the dictionary and ratio between the lowest to largest coefficients: 
\begin{align*}
\eta^2 \geq \min_{j} \frac{\left( \frac{|x_{\text{min,j}}|}{\|x_j\|_{\infty}} - \mu_1(S-1) - \mu_1(S) \right)^2}{S(1 + \mu_1(S-1))},
\end{align*}
where $\mu_1$ defines the cumulative coherence (Babel function) as defined in \cite{Tropp04} and $|x_{\text{min,j}}|$ is the absolute value of the smallest coefficient in vector $x_j$. Note that if $\Phi$ is an ONB, $\mu_1 = 0$. The proof of this inequality is very similar to the proof of Corollary 3.3 in \cite{Rauhut08}.

Another key parameter is the number of candidate vectors of transformations $|\mathcal{T}|$ which grows with $J$.
In the following corollary, we provide a lower bound on the number of measurements needed per sensor to reach asymptotically a perfect recovery rate in the following two cases: (1) $\mathcal{T}$ grows slowly with $J$, (2) $\mathcal{T}$ grows exponentially with $J$.

\begin{cor}[Asymptotic behaviour of $R$]
\label{corrolary}
Let $0 < \alpha \leq 1$.
\begin{enumerate}
	\item If $|\mathcal{T}|$ is a subexponential function of $J$, then, as long as $M \geq 1$, $\mathbb{P} (R \geq 1 - \alpha)$ converges to 1 as $J \rightarrow +\infty$.
	\item If there exists $\beta > 0$ such that $|\mathcal{T}| \sim e^{\beta J}$ , then, as long as $M > \frac{\beta}{C \eta^2 \alpha^2} \frac{M_y^2}{m_y^2}$, $\mathbb{P} (R \geq 1 - \alpha)$ converges to 1 as $J \rightarrow +\infty$.
\end{enumerate}
\end{cor}

\begin{proof}
The proof is a direct application of Theorem \ref{mainTheorem}.
\end{proof}

The growth of $|\mathcal{T}|$ is related to the degree of uncertainty on the correct transformation vector $T^*$. For example, if $T^*$ is known in advance, then $|\mathcal{T}| = 1$ and Corrolary \ref{corrolary} guarantees an arbitrary high recovery rate with only one measurement per sensor when $J \rightarrow +\infty$. This result remains valid as long as $|\mathcal{T}| \ll e^{\beta J}$ for all $\beta > 0$. However, if $T^*$ is completely unknown and transforms between pairs of signals are independent, $|\mathcal{T}|$ grows exponentially with $J$ and we will need more measurements per sensor in order to recover the correct support estimates (consider the example where (a) the number of candidate transformations between each sensor $j \geq 2$ and the reference signal is equal to $l$ ; (b) $T_j$ is independent of $T_{j-1}$, then: $|\mathcal{T}| = l^{J-1}$).

Unlike independent thresholding which has a constant recovery rate in function of $J$, previous results show that the recovery rate of JT \emph{increases} by augmenting the number of sensors $J$. Thus, in large networks, JT requires less measurements per sensor than independent thresholding for a fixed target recovery rate provided that $|\mathcal{T}|$ has a controlled growth.

\section{Experimental results}\label{sec:experiments}

\subsection{Greedy JT}

The JT algorithm, as described in Section \ref{sec:JT}, performs the search over all candidate transforms in $\mathcal{T}$. This can be very costly in terms of the computational efficiency, especially for a large number of correlated signals. Thus, instead of performing a full search, we greedily look for the relevant transformations. The Greedy Joint Thresholding (GJT) algorithm is described in Algorithm \ref{alg:gjt}.

\begin{algorithm}[ht]
\caption{Greedy Joint Thresholding (GJT) Algorithm}
\algorithmicrequire{\hspace{1mm} compressed signals $\{s_j\}$, sensing matrices $\{A_j\}$, sparsity $S$, dictionary $\Phi$, candidate vectors of transformations $\mathcal{T}$} \\
\algorithmicensure{\hspace{1mm} estimated signals $\{\widehat{y_j}\}$, support $\widehat{\Delta_1}$ and vector of transformations $\widehat{T}$}
\label{alg:gjt}
\begin{algorithmic}
\STATE $\mathbf{1.}$ Initialization: $\widehat{T} \leftarrow \mathbb{I}$ (identity).
\STATE $\mathbf{2.}$ For every $V \in \{2, \dots, J\}$
\STATE \hspace{2mm}$\mathbf{2.1}$ Set the values ($\widehat{\Delta_{1}}, \widehat{T_V}, \widehat{\Psi}) \leftarrow (\varnothing, \varnothing, -\infty)$.
\STATE \hspace{2mm}$\mathbf{2.2}$ Let $\mathcal{T}_{V}$ denote the possible transformations between signal 1 and signal $V$.
\STATE \hspace{2mm}$\mathbf{2.3}$ For each element $T_V \in \mathcal{T}_{V}$ do
\STATE \hspace{4mm}$\mathbf{2.3.1}$ Let $T \leftarrow [\widehat{T}, T_V]$ (\emph{i.e.,} $T_V$ appended to $\widehat{T}$)
\STATE \hspace{4mm}$\mathbf{2.3.2}$ Compute:
\begin{equation*}
	d_{T} = \sum\limits_{j=1}^V (A_j T_j (\Phi))^\mathrm{T} s_j.
\end{equation*}
\STATE \hspace{4mm}$\mathbf{2.3.3}$ Keep the largest S entries (set the other entries to zero). The positions of the non zero entries in $d_T$ give the indices of the estimated support $\Delta_1$.
\STATE \hspace{4mm}$\mathbf{2.3.4}$ Calculate the \emph{score} $\Psi_s(\Delta_1, T)$ by summing the $S$ non zero entries of $d_{T}$
\STATE \hspace{4mm}$\mathbf{2.3.5}$ If $\Psi_s(\Delta_1, T)$ exceeds $\widehat{\Psi}$: update $(\widehat{\Delta_1}, \widehat{T_V}, \widehat{\Psi}) \leftarrow (\Delta_1, T_V, \Psi_s(\Delta_1, T))$
\STATE \hspace{2mm}$\mathbf{2.4}$ Update the estimate of the vector of transformations: $\widehat{T} \leftarrow [\widehat{T}, \widehat{T_V}]$.
\STATE $\mathbf{3.}$ Perform steps 3. and 4. in Algorithm \ref{alg:jt}
\end{algorithmic}
\end{algorithm} 

Even though this algorithm has a lower complexity than JT, the price to pay is a less robust transformation estimation process: in the early stages of the algorithm ($V \ll J$), the selection of the transform is based on a small number of signals $V$. If in addition the value of $M$ is small, this may lead to uncorrect estimation of the transformations and thus wrong support estimates.

In Fig.\ref{fig:greedy_fullsearch}, we plot the percentage of uncorrect estimated transforms with JT and Greedy JT in function of $M$, for a randomly generated image. For $M \geq 80$, the performance loss with Greedy JT is relatively small with respect to the gain in complexity. 
\begin{figure}[h!]
\centering
\includegraphics[width=0.4\columnwidth]{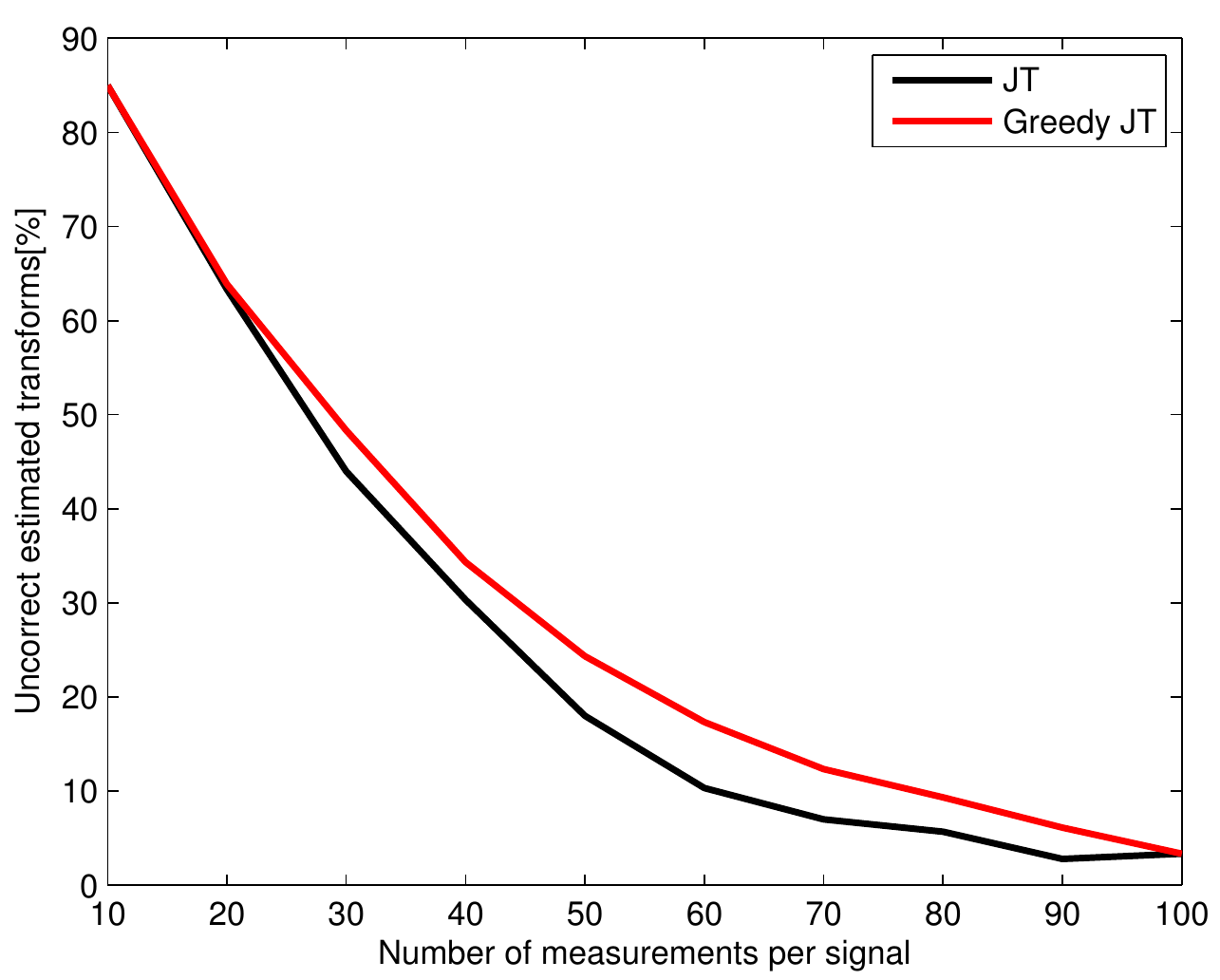}
\caption{Transforms estimations using JT and Greedy JT. Simulation setup: 20 independent trials, $J=4$, $S=5$, $N=32 \times 32$, Gaussian sensing matrices, independent transformations and $|\mathcal{T}| = 9^3$. The used dictionary given in section \ref{randomGeneratedImages}.}
\label{fig:greedy_fullsearch}
\end{figure}
Thus the penalty of using the greedy algorithm is small in practice.
In the following, we examine the performance of Greedy JT on synthetic images and seismic signals.




\subsection{Synthetic images}
\label{randomGeneratedImages}
We construct a parametric dictionary where a generating function undergoes rotation, scaling and translation operations to generate the different atoms in the dictionary $\Phi$. We use the Gaussian $g(x,y) = e^{-x^2 -y^2}$ as the generating function. The atoms in the dictionary are characterized by the rotation angle $\theta$, scales $s_x$ and $s_y$ and translations $t_x$ and $t_y$. If $(X,Y)$ denotes the transformed coordinate system:
\begin{align*}
X & = \frac{ (x-t_x) \cos \theta - (y-t_y) \sin \theta }{s_x} \\
Y & = \frac{ (y-t_y) \cos \theta + (x-t_x) \sin \theta }{s_y},
\end{align*}
the atom $g_p$ with parameters $p = (\theta, s_x, s_y, t_x, t_y)$ is given by:
\begin{equation*}
g_{p} (x, y) = \rho g (X, Y), 
\end{equation*}
where $\rho$ is the normalization constant.

The dictionary is generated for images of size $N = 32 \times 32 = 1024$, with the following parameters: $\theta \in [0:\frac{\pi}{6}:\pi], s_x = \{2, 4\}, s_y = \{1/2, 1\}$. Every atom is shifted in pixels of odd coordinates, so the full dictionary contains 6144 atoms.

The support of the reference image and coefficients are chosen in order to verify the conditions in Eq.(\ref{assumptionThresholding}) and Eq.(\ref{positivity}). The remaining images have been obtained by applying global translations on the atoms of the reference image, under the constraint that all atoms in images belong to $\mathcal{D}$. We assume that the transformations are independent from one another and that there are 9 candidate transformations for any image. Thus, $|\mathcal{T}| = 9^{J-1}$. 
\begin{figure}[tbh]
\centering
\includegraphics[width=0.6\columnwidth]{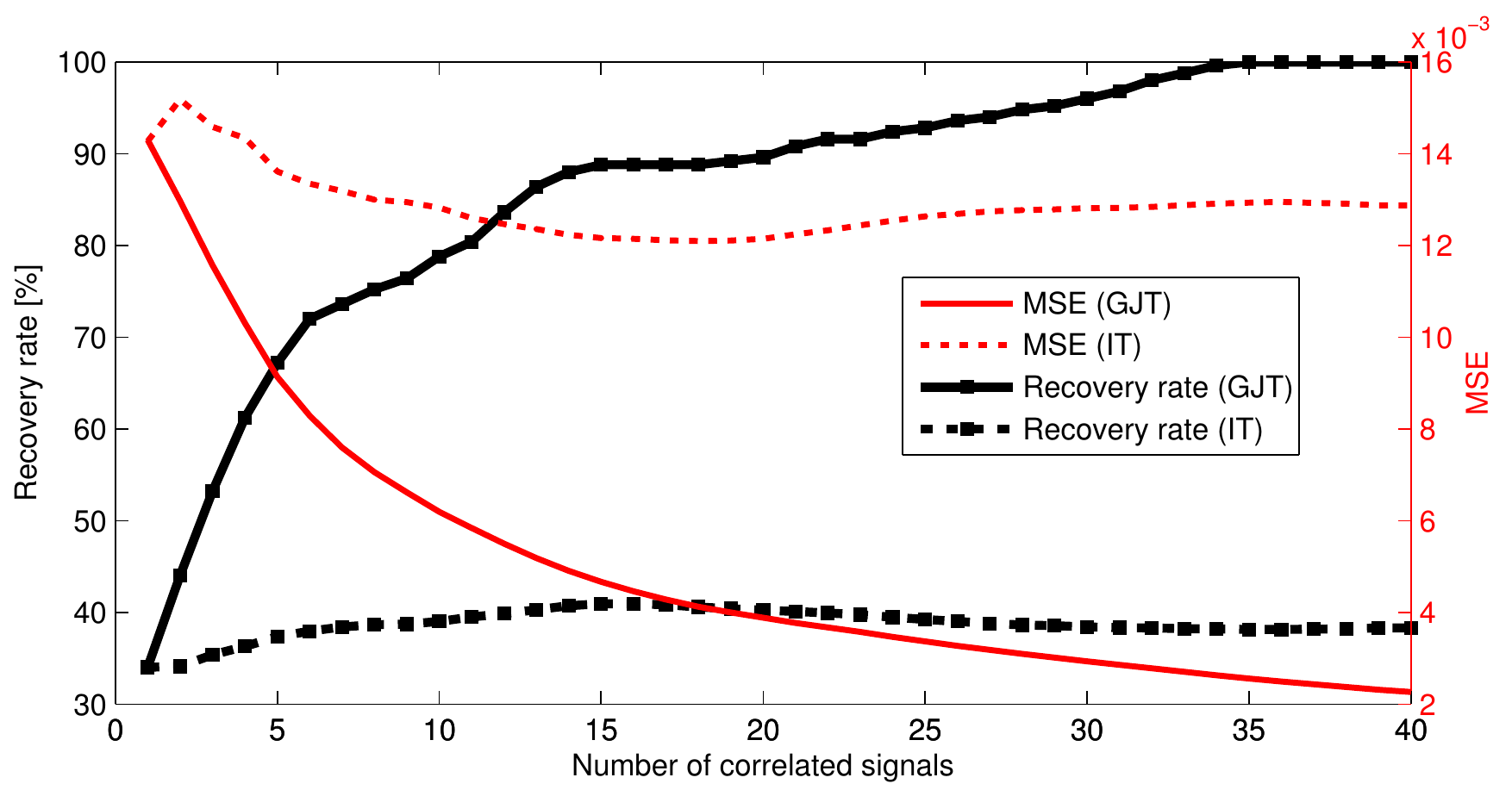}
\caption{\footnotesize{Recovery rate and Mean Squared Error of Greedy JT (GJT) and Independent thresholding (IT) in function of $J$. Simulation setup: 10 independent trials, $S = 5$, $M = 150$, $N = 1024$, Gaussian sensing matrices.}}
\label{fig:performance}
\vspace{-2mm}
\end{figure}
Fig.\ref{fig:performance} illustrates the recovery rate and MSE of Greedy JT and independent thresholding for a randomly generated image. Recovery rate is defined in Section \ref{quantitative}. For a given $J$, the calculated MSE represents the averaged MSE calculated over signals $\{1, \dots, J\}$. We see that Greedy JT outperforms independent thresholding in terms of recovery rate and image quality, especially for high values of $J$ ($J \geq 20$). Thus, although $|\mathcal{T}|$ grows rapidly with $J$, our joint decoding approach is significantly better in practice in terms of support recovery.
\subsection{1D seismic signals}

Seismic signals captured at neighbouring locations typically follow the correlation model proposed in this paper. Fig.\ref{fig:seismic} (a), (b) represent two seismic signals that are obviously correlated as the second signal is approximately a shifted version toward the front of the first signal. 
We use the following sparsifying dictionary, which consists of Gaussians modulated with sinusoids:
\begin{equation*}
g_{(t, s, \omega)}(x) = \rho \exp\left(-\frac{(x-t)^2}{s^2}\right) \cos\left(\omega \frac{x - t}{s}\right),
\end{equation*}
where $\rho$ is the normalizing constant.
The translations $t$ are chosen uniformly from $1$ to $N$ with step size 10 such that the coherence of the dictionary is not too high. Scales $s$ take values in $\{4, 8, 16\}$ and $\omega$ varies from $2$ to $10$ with step 2. For each set of parameters $(t, s, \omega)$, $g_{(t, s, \omega)}$ and $-g_{(t, s, \omega)}$ are included in the dictionary. Fig.\ref{fig:seismic} (c) and (d) illustrate the estimations of signal number 2 obtained with only $15\%$ of the measurements respectively using independent thresholding and JT algorithm. Note that as $J=2$ in this example, Greedy JT and JT are equivalent. Visual inspection and calculated MSEs confirm the superiority of joint decoding using JT algorithm over independent thresholding in terms of reconstruction quality. 
\begin{figure}[tbh]
\centering
\includegraphics[width=0.5\columnwidth]{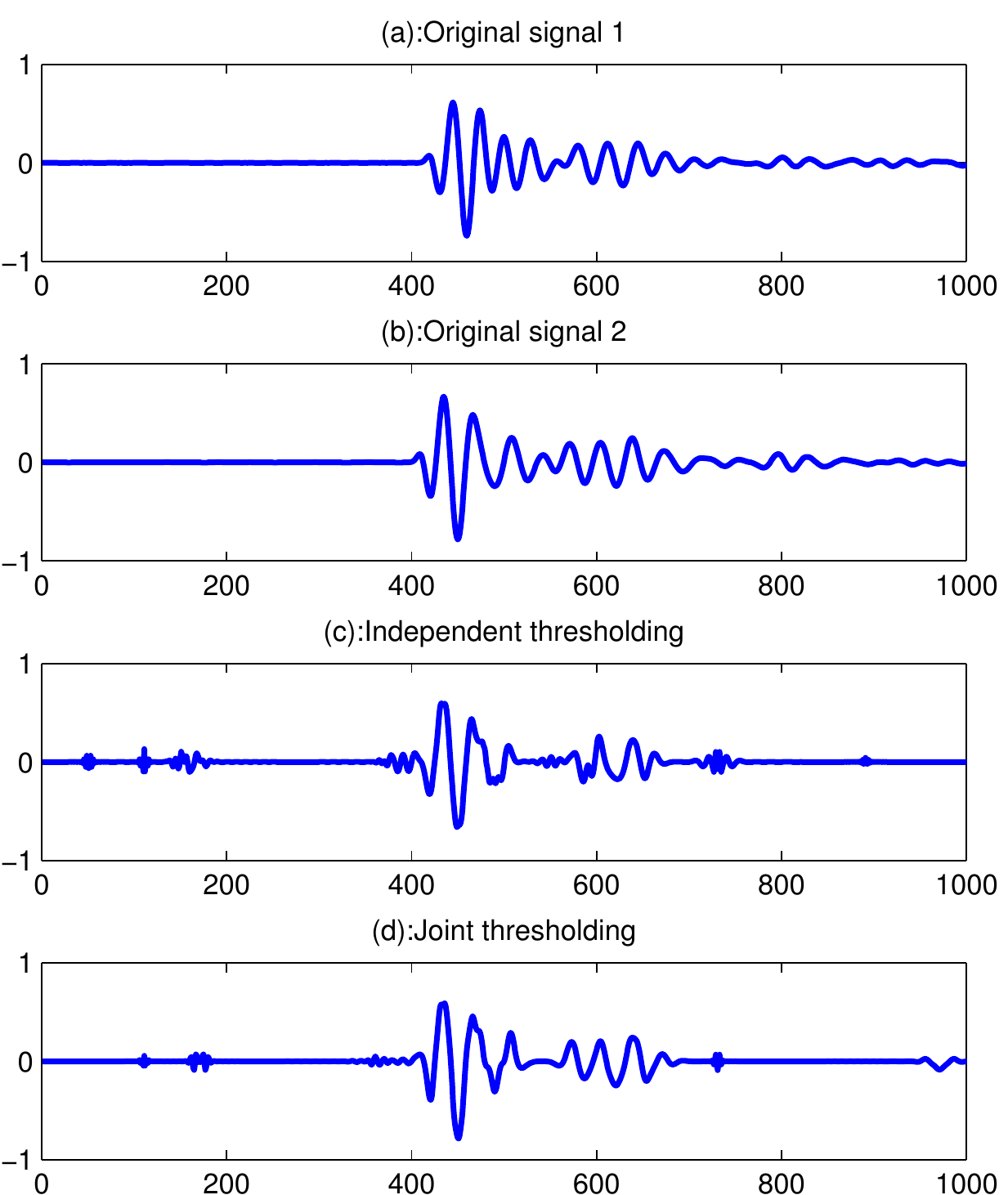}
\caption{\footnotesize{Seismic signals (a) $y_1$ and (b) $y_2$ captured at two neighbouring locations. Estimation of $y_2$ using (c) independent thresholding and (d) JT. Simulation setup: $J=2, N=1000, M=150, S=50$, $|\mathcal{T}| = 3$, Gaussian sensing matrices. This experiment was conducted 200 times and we obtained $\text{MSE}_{\text{IT}} = 0.0031$ and $\text{MSE}_{\text{JT}} = 0.0025$}.}
\label{fig:seismic}
\vspace{-2mm}
\end{figure}
This experiment shows that joint decoding using JT provides significantly better quality signals even when the number of correlated signals is low ($J=2$).
\section{Conclusion}\label{sec:conclusion}

In this paper, we have proposed an efficient approach for the joint recovery of correlated signals that have been compressed independently.
Our solution is novel with respect to the state of the art work due to the particular geometrical correlation model based on the transformations of the sparse signal components. 
Mathematical analysis and experimental results demonstrate the superiority of our recovery algorithm over independent thresholding. JT is namely applicable for decoding simple multiview images, seismic signals or any other set of correlated signals satisfying the geometric correlation model. A promising future direction is to use JT for correlation estimation along with a more sophisticated recovery algorithm for the reconstruction.
\section{Acknowledgments}\label{sec:acknowledgments}

The first author would like to thank Omar Fawzi for the fruitful discussions.

\clearpage
\addcontentsline{toc}{section}{Appendix}
\section*{Appendix A: Proof of Lemma \ref{fundLemma}}

This proof is inspired from the proof of \cite[Lemma 3.1]{Rauhut08}.

Let $(g_j)_{m, n}$ be a random variable following the standard gaussian distribution such that $(A_j)_{m, n} = \frac{1}{\sqrt{M}} (g_j)_{m, n}$.
We have, for any set of vectors $u_j, v_j$ in $\mathbb{R}^N$:
\begin{equation*}
\frac{1}{J} \sum\limits_{j=1}^J A_ju_j \cdot A_jv_j = \frac{1}{MJ} \sum\limits_{j=1}^J \sum\limits_{m=1}^M \sum\limits_{k=1}^N \sum\limits_{l=1}^N (g_j)_{m, k} (g_j)_{m, l} (u_j)_k (v_j)_l
\end{equation*}
Let $Y_{j, m} = \sum\limits_{k=1}^N \sum\limits_{l=1}^N (g_j)_{m, k} (g_j)_{m, l} (u_j)_k (v_j)_l$.
As $(g_j)_{m, k}$ is independent from $(g_j)_{m, l}$, we have $\mathbb{E}Y_{j,m} = u_j \cdot v_j$.
Let $Z_{j, m} = Y_{j, m} - \mathbb{E}Y_{j,m} = \sum\limits_{k \neq l} (g_j)_{m, k} (g_j)_{m, l} (u_j)_k (v_j)_l + \sum\limits_{k} (\left((g_j)_{m, k}\right)^2 - 1) (u_j)_k (v_j)_k$.
By definition, for any $j$ and $m$, $Z_{j, m}$ is a Gaussian chaos of order 2.

Observe that the probability we wish to bound can be expressed in terms of $Z_{j, m}$:
\begin{align*}
\mathbb{P} \left(\frac{1}{J} \left| \sum\limits_{j = 1}^J A_j u_j \cdot A_j v_j - u_j \cdot v_j \right| \geq \tau \right)
& = \mathbb{P} \left(\frac{1}{J} \left| \sum\limits_{j = 1}^J \frac{1}{M} \sum\limits_{m=1}^M \left( Y_{j, m} - \mathbb{E} Y_{j, m} \right) \right| \geq \tau \right) \\ &
= \mathbb{P} \left(\left| \sum\limits_{j = 1}^J \sum\limits_{m=1}^M Z_{j, m} \right| \geq \tau M J \right).
\end{align*}

As $Z_{j, m}$ is a Gaussian chaos of order 2, Bernstein's inequality on the sum of zero mean independent random variables with a certain moment growth is applicable (notice that the independence assumption is satisfied in our case as entries are iid and $\{A_j\}_{j=1}^J$ are independent). For more details about the theorem refer to Theorem A.1 in \cite{Rauhut08}.
We get:
\begin{align}
\label{bernstein}
\mathbb{P} \left(\left| \sum\limits_{j = 1}^J \sum\limits_{m=1}^M Z_{j, m} \right| \geq \tau M J \right) \leq 2 \exp \left( {-\frac{1}{2} \frac{J M \tau^2}{w + z \tau}} \right).
\end{align}
with $w = \max_{j, l} \mathbb{E} [Z_{j, l}^2] \frac{2e}{\sqrt{6 \pi}}$ and $z = e \sqrt{(\max_{j,l} \mathbb{E} [Z_{j, l}^2])}$.

By expanding $Z_{j, l}^2$, we calculate $\mathbb{E}[ Z_{j, l}^2 ]$ and obtain:
\begin{equation*}
\mathbb{E} [Z_{j, l}^2] = \|u_j\|_2^2 \|v_j\|_2^2 + (u_j \cdot v_j)^2 \leq 2 \|u_j\|_2^2 \|v_j\|_2^2 \leq 2 B_u^2 B_v^2
\end{equation*}
Thus, $w \leq \frac{4e}{\sqrt{6 \pi}} B_u^2 B_v^2$ and $z = \sqrt{2} e B_u B_v$.
We finally obtain the desired result by replacing the expressions of $w$ and $z$ in  Eq.(\ref{bernstein}).

\section*{Appendix B: Proof of Theorem \ref{mainTheorem}}

The proof is composed of 2 steps.

\begin{enumerate}
\item We first perform a simple calculation that will be needed in the second part of the proof.

Let $\{\widehat{\Delta_j}\}_{j=1}^J$ denote an estimated set of supports having $k$ incorrect atoms:
\begin{equation*}
\sum\limits_{j=1}^J |\Delta_j^* \cap \overline{\widehat{\Delta_j}}| = k.
\end{equation*} 
We have the equality:
\begin{align*}
& \sum\limits_{j=1}^J \sum_{ \varphi \in \Delta_j^* } y_j \cdot \varphi - \sum\limits_{j=1}^J \sum_{ \varphi \in \widehat{\Delta_j} } y_j \cdot \varphi = \sum\limits_{j=1}^J \sum_{ \varphi \in \Delta_j^* \cap \overline{\widehat{\Delta_j}} } y_j \cdot \varphi - \sum\limits_{j=1}^J \sum_{ \varphi \in \widehat{\Delta_j} \cap \overline{\Delta_j^*} } y_j \cdot \varphi,
\end{align*}
as correct atoms in the estimated supports cancel out.
Condition in Eq.(\ref{assumptionThresholding}), together with the positivity assumption in Eq.(\ref{positivity}) imply that:
\begin{equation}
\sum\limits_{j=1}^J \sum_{ \varphi \in \Delta_j^* } y_j \cdot \varphi - \sum\limits_{j=1}^J \sum_{ \varphi \in \widehat{\Delta_j} } y_j \cdot \varphi > \eta k m_y
\end{equation}
\item
For a fixed vector of transformations $T$, we define the set of supports $\{\Delta_j^T\}_{j=1}^J$ in the following way:
\begin{itemize}
	\item $\Delta_1^T$ as the set of $S$ atoms maximizing $d_T[i] = \sum_{j=1}^J s_j \cdot A_j T_j(\varphi_i)$ .
	\item For $j \geq 2$, $\Delta_j^T = T_j(\Delta_1^T)$
\end{itemize}

By definition, the support $\widehat{\Delta_1}$ is composed of the $S$ atoms maximizing $d_{\widehat{T}}$. Hence, $\widehat{\Delta_1} = \Delta_1^{\widehat{T}}$.


We say that the estimated supports are \emph{$h$-incorrect}, if the total number of incorrectly estimated atoms (in all signals combined) is at least equal to $h$:
\begin{equation*}
\sum\limits_{j=1}^J |\Delta_j^* \cap \overline{\widehat{\Delta_j}}| \geq h,
\end{equation*}
We consider the event $\{$ JT estimates $h-$incorrect supports $\}$ and we write the following equalities and inclusions on the events:
\begin{align}
\left\{ \text{JT estimates } h\text{-incorrect supports} \right\}
& = \left\{ (\widehat{\Delta_1}, \widehat{T})  = \argmax_{(\Delta_1, T)} \Psi_s(\Delta_1, T) \text{ verifies } \sum\limits_{j=1}^J \left|\Delta_j^* \cap \overline{\widehat{\Delta_j}}\right| \geq h \right\} \nonumber \\ & 
= \left\{ \widehat{T}  = \argmax_{T \in \mathcal{T}} \Psi_s(\Delta_1^T, T) \text{ verifies } \sum\limits_{j=1}^J \left| \Delta_j^* \cap \overline{\Delta_j^{\widehat{T}}} \right| \geq h \right\} \nonumber \\ &
\subset  \left\{ \exists T \in \mathcal{T} \text{ verifying } \Psi_s(\Delta_1^*, T^*) \leq \Psi_s(\Delta_1^T, T) \text{ and } \sum\limits_{j=1}^J \left| \Delta_j^* \cap \overline{\Delta_j^{T}} \right| \geq h \right\} \nonumber \\ &
\subset \bigcup_{T \in \mathcal{T}} \left\{ \Psi_s (\Delta_1^*, T^*) \leq \Psi_s(\Delta_1^T, T), \sum\limits_{j=1}^J  \left| \Delta_j^* \cap \overline{\Delta_j^{T}} \right| \geq h \right\}.
\end{align}

Thus,
\begin{align*}
\mathbb{P} \left(\text{JT estimates $h$-incorrect supports} \right)
& \leq \mathbb{P} \left( \bigcup_{T \in \mathcal{T}} \left\{ \Psi_s (\Delta_1^*, T^*) \leq \Psi_s (\Delta_1^T, T), \sum\limits_{j=1}^J  \left| \Delta_j^* \cap \overline{\Delta_j^T} \right| \geq h \right\} \right) \\
& \leq \sum\limits_{T \in \mathcal{T}} \mathbb{P} \left( \Psi_s (\Delta_1^*, T^*) \leq \Psi_s (\Delta_1^T, T), \sum\limits_{j=1}^J  \left|\Delta_j^* \cap \overline{\Delta_j^T} \right| \geq h \right) \\
& \leq \sum\limits_{T \in \mathcal{T}} \sum\limits_{k=\lceil h \rceil}^{SJ} \mathbb{P} \left( \Psi_s (\Delta_1^*, T^*) \leq \Psi_s (\Delta_1^T, T), \sum\limits_{j=1}^J | \Delta_j^* \cap  \overline{\Delta_j^T} | = k \right). \\
\end{align*}
From the first part of the proof, we know that if $\sum\limits_{j=1}^J \left| \Delta_j^* \cap  \overline{\Delta_j^T} \right| = k$, then, $\Psi_y(\Delta_1^*, T^*) - \Psi_y(\Delta_1^T, T) > \eta k m_y$, where $\Psi_y$ is defined in Section 3. Hence, the following inclusion holds:
\begin{align*}
& \left\{ \Psi_s (\Delta_1^*, T^*) \leq \Psi_s (\Delta_1^T, T), \sum\limits_{j=1}^J \left| \Delta_j^* \cap  \overline{\Delta_j^T} \right| = k \right\} \\  \subset & \left\{ \Psi_s (\Delta_1^*, T^*) \leq \Psi_y (\Delta_1^*, T^*) - \frac{\eta k m_y}{2} \right\} \bigcup \left\{ \Psi_s (\Delta_1^T, T) \geq \Psi_y (\Delta_1^T, T) + \frac{\eta k m_y}{2} \right\}.
\end{align*}

Hence,
\begin{align}
\mathbb{P} & \left( \Psi_s (\Delta_1^*, T^*) \leq \Psi_s (\Delta_1^T, T),  \sum\limits_{j=1}^J \left| \Delta_j^* \cap  \overline{\Delta_j^T} \right| = k \right) \label{pALL} \\ \leq
\mathbb{P} & \left( \sum\limits_{j=1}^J \sum_{ \varphi \in \Delta_j^* } s_j \cdot A_j \varphi \leq \sum\limits_{j=1}^J \sum_{ \varphi \in \Delta_j^* } y_j \cdot \varphi - \frac{\eta k m_y }{2} \right) \nonumber +
\mathbb{P} \left( \sum\limits_{j=1}^J \sum_{ \varphi \in \Delta_j^T } s_j \cdot A_j \varphi \geq \sum\limits_{j=1}^J \sum_{ \varphi \in \Delta_j^T } y_j \cdot \varphi + \frac{\eta k m_y}{2} \right) \nonumber \\ \leq
\mathbb{P} & \left( \exists \varphi \in \Delta_1^*, \sum\limits_{j=1}^J y_j \cdot T_j^* (\varphi) - s_j \cdot A_j T_j^*(\varphi) \geq \frac{\eta k m_y}{2S} \right) \label{p1}  \\ +
\mathbb{P} & \left( \exists \varphi \in \Delta_1^T, \sum\limits_{j=1}^J s_j \cdot A_j T_j(\varphi) - y_j \cdot T_j(\varphi) \geq \frac{\eta k m_y}{2S} \right) \label{p2}
\end{align}

We rewrite probability in Eq.(\ref{p1}) and apply Lemma \ref{fundLemma}:
\begin{align*}
\mathbb{P} \left( \bigcup_{\varphi \in \Delta_1^*} \left\{ \sum\limits_{j=1}^J y_j \cdot T_j^*(\varphi) - s_j \cdot A_j T_j^*(\varphi) \geq \frac{\eta k m_y}{2S} \right\} \right) &
\leq \sum_{\varphi \in \mathcal{D}} \mathbb{P} \left( \frac{1}{J} \left| \sum\limits_{j=1}^J  s_j \cdot A_j T_j^*(\varphi) - y_j \cdot T_j^*(\varphi) \right| \geq \frac{\eta k m_y}{2SJ} \right) \\ &
\leq 2 K \exp {\left( - \frac{M k^2 \eta^2}{4C_1 S^2 J + 2C_2\eta k S} \frac{m_y^2}{M_y^2} \right)}
\end{align*}

Similarly for Eq.(\ref{p2}),
\begin{align*}
\mathbb{P} \left( \bigcup_{\varphi \in \Delta_1^T} \left\{ \sum\limits_{j=1}^J s_j \cdot A_j T_j(\varphi) - y_j \cdot T_j(\varphi) \geq \frac{\eta k m_y}{2S} \right\} \right) &
\leq \sum_{\varphi \in \mathcal{D}} \mathbb{P} \left( \frac{1}{J} \left| \sum\limits_{j=1}^J  s_j \cdot A_j T_j(\varphi) - y_j \cdot T_j(\varphi) \right| \geq \frac{\eta k m_y}{2SJ} \right) \\ &
\leq 2 K \exp {\left( - \frac{M k^2 \eta^2}{4C_1 S^2 J + 2C_2\eta k S} \frac{m_y^2}{M_y^2} \right)}
\end{align*}
Thus, by combining both results, we obtain:
\begin{align*}
\mathbb{P} \left(\text{JT estimates } h\text{-incorrect supports} \right) & \leq \sum\limits_{T \in \mathcal{T}} \sum\limits_{k=\lceil h \rceil}^{SJ} 4K \exp {\left( - \frac{M k^2 \eta^2}{4C_1 S^2 J + 2C_2\eta k S} \frac{m_y^2}{M_y^2} \right)} \\ & \leq
4 SJK |\mathcal{T}| \exp {\left( - \frac{M h^2 \eta^2}{4C_1 S^2 J + 2C_2\eta h S} \frac{m_y^2}{M_y^2} \right)}
\end{align*}

\end{enumerate}
Let $h = \alpha S J$, then, we have:
\begin{align*}
\mathbb{P} \left(\text{JT estimates } (\alpha S J)\text{-incorrect supports} \right) & = \mathbb{P} \left( \frac{\text{Number of errors}}{SJ} \geq \alpha \right) \\ & = \mathbb{P} (R \leq 1 - \alpha) \\ & \leq 4 SJK |\mathcal{T}| \exp {\left( - \frac{M J \alpha^2 \eta^2}{4C_1 + 2C_2 \eta \alpha } \frac{m_y^2}{M_y^2} \right)}
\\ & \leq 4 SJK |\mathcal{T}| \exp {\left( - C M J \alpha^2 \eta^2 \frac{m_y^2}{M_y^2} \right)}
\end{align*}
with $C = (4C_1 + 2C_2)^{-1} = \left(\frac{32 e}{\sqrt{6\pi}} + 4 e \sqrt{2} \right)^{-1}$.
The result of Theorem \ref{mainTheorem} is finally obtained by taking the probability on the complementary event.
\newpage
\bibliographystyle{IEEEbib}
\bibliography{refs_technicalReport}

\end{document}